\documentclass[a4paper,USenglish, autoref, cleveref, thm-restate, numberwithinsect, notab]{lipics-v2021}

\usepackage{tcolorbox}
\usepackage{xspace}
\usepackage{xcolor}
\usepackage[ruled]{algorithm}
\usepackage{algpseudocode}

\title{Decomposing a Simple Polygon with Geodesic Unit-Balls} %TODO Please add

\author{Reilly Browne}{Department of Computer Science, Dartmouth College, Hanover, NH, USA}{reilly.browne.gr@dartmouth.edu}{https://orcid.org/0000-0003-3725-5245}{}

\author{Prahlad Narasimhan Kasthurirangan}{Department of Applied Mathematics and Statistics, Stony Brook University, USA}{prahladnarasim.kasthurirangan@stonybrook.edu}{https://orcid.org/0000-0002-8518-7745}{Supported by the US Office of Naval Research under grant no. N00014-26-1-2088.}

\authorrunning{R. Browne and P.N. Kasthurirangan} %TODO mandatory. First: Use abbreviated first/middle names. Second (only in severe cases): Use first author plus 'et al.'

\Copyright{Reilly Browne and Prahlad Narasimhan Kasthurirangan} %TODO mandatory, please use full first names. LIPIcs license is "CC-BY";  http://creativecommons.org/licenses/by/3.0/


 %TODO mandatory; please add comma-separated list of keywords

\category{} %optional, e.g. invited paper

\relatedversion{} %optional, e.g. full version hosted on arXiv, HAL, or other respository/website

\acknowledgements{We would like to thank the anonymous reviewers of ESA 2026 for their valuable feedback.}%optional

\EventEditors{Philip Bille, Seth Pettie, and Sabine Storandt}
\EventNoEds{3}
\EventLongTitle{34th Annual European Symposium on Algorithms (ESA 2026)}
\EventShortTitle{ESA 2026}
\EventAcronym{ESA}
\EventYear{2026}
\EventDate{August 31--September 4, 2026}
\EventLocation{L'Aquila, Italy}
\EventLogo{}
\SeriesVolume{388}
\ArticleNo{65}

\newcommand{\Q}{\mathcal{Q}}
\renewcommand{\S}{\mathcal{S}}
\newcommand{\I}{\mathcal{I}}

\newcommand{\reals}{\mathbb{R}}
\newcommand{\naturals}{\mathbb{N}}

\newcommand{\order}[2][]{\mathcal{O}_{#1}(#2)}
\newcommand{\ordertilde}[2][]{\widetilde{\mathcal{O}}_{#1}(#2)}
\newcommand{\nph}{\textup{\textsc{NP-Hard}}\xspace}

\newcommand{\false}{\textup{\textsc{false}}\xspace}
\newcommand{\true}{\textup{\textsc{true}}\xspace}

\newcommand{\opt}{\textup{\textsc{opt}}\xspace}

\newcommand{\dg}[3][P]{d_{#1}(#2, #3)}
\newcommand{\de}[2]{d(#1, #2)}
\newcommand{\spath}[3][P]{\pi_{#1}(#2, #3)}
\newcommand{\balle}[2][1]{B_{#1}(#2)}
\newcommand{\ballg}[2][P]{B_{1}^{#1}(#2)}
\newcommand{\ballgr}[3][P]{B^{#1}_{#2}(#3)}
\newcommand{\sector}[3][c]{S(#1,#2,#3)}
\newcommand{\sextant}[2][c]{S'(#1,#2)}

\newcommand{\GRCover}{\textup{\textsc{Small GR-Cover}}\xspace}
\newcommand{\GRPart}{\textup{\textsc{Small GR-Partition}}\xspace}

\newcommand{\wordRGeodesicCover}{small geodesic radius cover\xspace}
\newcommand{\wordGRPartition}{small GR-partition\xspace}
\newcommand{\wordGRCover}{small GR-cover\xspace}
\newcommand{\GRpiece}{small GR-piece\xspace}
\newcommand{\GeodesicRpiece}{small geodesic radius piece\xspace}
\newcommand{\GRgreedyboundaryalgo}{\textup{\textsc{SmallGRGreedyBoundaryCover}}\xspace}
\newcommand{\greedyinterioralgo}{\hyperref[alg:greedy_interior]{\textup{\textsc{GreedyInteriorCover}}}\xspace}
\newcommand{\GRextracentersalgo}{\hyperref[alg:extra_centers]{\textup{\textsc{SmallGRAddRedundantCenters}}}\xspace}
\newcommand{\GRgreedypolygonalgo}{\hyperref[alg:greedy_full]{\textup{\textsc{SmallGRGreedyCover}}}\xspace}

\newcommand{\GDCover}{\textup{\textsc{Small GD-Cover}}\xspace}
\newcommand{\GDPart}{\textup{\textsc{Small GD-Partition}}\xspace}

\newcommand{\wordGDPartition}{small GD-partition\xspace}

\newcommand{\wordGDCover}{small GD-cover\xspace}
\newcommand{\GDpiece}{small GD-piece\xspace}
\newcommand{\GeodesicDpiece}{small geodesic diameter piece\xspace}
\newcommand{\GDgreedyboundaryalgo}{\textup{\textsc{SmallGDGreedyBoundaryCover}}\xspace}
\newcommand{\GDextracentersalgo}{\hyperref[alg:GD_extra_centers]{\textup{\textsc{SmallGDAddRedundantCenters}}}\xspace}
\newcommand{\GDgreedypolygonalgo}{\hyperref[alg:GD_greedy_full]{\textup{\textsc{SmallGDGreedyPartition}}}\xspace}

\theoremstyle{remark}
\newtheorem{innercase}{Case}[theorem]

\DeclareRobustCommand{\lipicsEnd}{%
	\leavevmode\unskip\penalty9999 \hbox{}\nobreak\hfill
	\quad\hbox{$\lrcorner$}%
}

\newenvironment{case}
{%
  \pushQED{\lipicsEnd}%
  \begin{innercase}%
}
{%
  \popQED%
  \end{innercase}%
}

\nolinenumbers

\begin{document}

\maketitle

\begin{abstract}
    We consider covering and partitioning a simple polygon into pieces which either have unit geodesic radius or unit geodesic diameter, using the $\ell_2$-metric for distances. There is no known method for finding an exact solution to these problems, even when the input size is constant, and the problem is known to be NP-hard in the case of polygons with holes. With this in mind, we instead devote our attention to developing simple approximation algorithms that run in polynomial time. For the radius problem, we present the first known approximation algorithms for both covering and partitioning, achieving a factor of 9. For the diameter problem, we are only able to give a positive result for the partition version of the problem, where we improve upon a complicated 72-approximation from Abrahamsen and Rasmussen \cite{abrahamsen_partitioning_2022}, achieving a simple 15-approximation.
\end{abstract}

\section{Introduction}
%Story: like Vigan \cite{vigan_packing_2013}. Equitable partitioning studied. Two ideas here. (1) Want to ensure that nobody  has to walk more than 1 unit within polygon to get to their nearest center in a polygon (radius; covering with disks). (2) No two points in the same cell are more than 1 unit apart (diameter). Covering points has been studied extensively, but many practical problems require covering continuum. Similar to $k$-center in the continuum, studied extensively (although not so much in polygons). Covering boundary has been looked at by Vigan. (2) has been looked by Mikkel. There has also been recent work on decomposing polygons into other types of small pieces.  

%\begin{itemize}
%    \item General partitioning: \cite{keil_decomposing_1985,mark_keil_chapter_2000}
%    \item Equitable/balanced partitioning: \cite{carlsson_dividing_2013,behroozi_computational_2020,kostitsyna_balanced_nodate}
%    \item $k$-center, recently in polygons \cite{suzuki_p-center_1996,du_approximation_2014,oh_geodesic_2018,evans_polygon_2022,de_berg_clustering_2023}
%    \item Covering points: \cite{fowler_optimal_1981,gonzalez_covering_1991,hochbaum_approximation_1985,biniaz_approximation_2017} \cite{browne_single_2026} \prahlad{Add your SoCG paper, Reilly?}
%\end{itemize}

A fundamental problem which appears in many domains is that of describing large regions as a series of smaller, easier to process pieces. Depending on the application, there are varying definitions of small which are relevant. In some settings, we want pieces which have some simple property to work with, such as convex pieces or star-shaped pieces \cite{keil_decomposing_1985,mark_keil_chapter_2000}, while in others, we care more about physical size, such as whether a piece fits inside of a square \cite{abrahamsen_partitioning_2022,abrahamsen_hardness_2024,aamand_covering_2026} or that this division is, in some sense, equitable  \cite{carlsson_dividing_2013,kostitsyna_balanced_2013,behroozi_computational_2020}. There is also the concern of whether the pieces are allowed to overlap, in which case we are finding a \emph{cover} of the region, or whether the pieces must be (interior) disjoint, in which case we are looking for a \emph{partition} and whether we are covering a continuous domain or a discrete set of points \cite{fowler_optimal_1981,gonzalez_covering_1991,hochbaum_approximation_1985,biniaz_approximation_2017}.

Our focus will be on two notions of \emph{geodesic} smallness, specifically pieces which have unit geodesic \emph{radius} or unit geodesic \emph{diameter} (see Section~\ref{sec:prelims} for formal definitions), and our region of interest will be a simple polygon $P$. 

In the radius problem, we are finding a set of center points such that every point in $P$ is within unit geodesic distance of at least one center. This notion of ``small'' is natural for facility location applications and is equivalent to the dual of the famous $k$-\textsc{center} problem \cite{suzuki_p-center_1996,du_approximation_2014,oh_geodesic_2018,evans_polygon_2022,de_berg_clustering_2023}: In $k$-\textsc{center}, we have a fixed budget of $k$ centers and want to instead minimize the maximum distance from a point to its nearest center. The problem of covering a \textit{discrete} set of points in $P$ with geodesic unit-balls admits a local search--based PTAS~\cite{mustafa-irghs-2010,harpeled-aaplg-2017a}, owing to the existence of planar supports~\cite{browne_single_2026}. In contrast, for covering (a continuum in) $P$, progress has been limited. Rabanca and Vigan gave a $2$-approximation algorithm for covering \textit{the boundary} of $P$ in 2015~\cite{rabanca_covering_2015}, and we are unaware of any subsequent improvements. Indeed, designing an algorithm to cover all of $P$ was explicitly posed as an open problem by Abrahamsen and Rasmussen in 2022~\cite{abrahamsen_partitioning_2022}.

The state of the art for the diameter variant, in which no two points within the same piece are at geodesic distance greater than one, is somewhat stronger. While no exact algorithms are known (even with exponential running time), Abrahamsen and Rasmussen achieve a $72$-approximation for the partition version~\cite{abrahamsen_partitioning_2022}. 

\subsection{Preliminaries}\label{sec:prelims}

Consider two points $p$ and $q$ on the plane. The \textit{Euclidean distance} between them is denoted by $\de{p}{q}$. For $r \in \reals$, let $\balle[r]{p}$ denote the \textit{Euclidean ball} of radius $r$ centered at $p$, i.e., $\balle[r]{p} := \{q \in \reals^2 \mid \de{p}{q} \leq r\}$.

Throughout this paper, we use the term \textit{polygonal domain} to denote polygons (possibly) with holes and \textit{polygons} for simple polygons. For two points $p$ and $q$ in a polygonal domain $P$, let $\spath{p}{q}$ denote a shortest path\footnote{Shortest paths between points within a polygon are unique; this might not hold in a polygonal domain.} between them within $P$. The \textit{geodesic distance} $\dg{p}{q}$ is defined as $|\spath{p}{q}|$, the length of a shortest path between $p$ and $q$ within $P$ . The \textit{geodesic ball} of radius $r$ centered at $p$ is defined as $\ballgr{r}{p} := \{q \in P \mid \dg{p}{q} \leq r\}$. The \textit{geodesic radius} of a polygonal domain $P$ is the smallest $r$ such that $\ballgr{r}{c} = P$ for some $c \in P$.
The \textit{geodesic diameter} of $P$ is the maximum geodesic distance between any two points in $P$.

A polygon $Q$ is a \textit{\GeodesicRpiece} (hereafter, \textit{\GRpiece}) if the geodesic radius of $Q$ is at most $1$. A collection of polygons $\Q$ is a \textit{\wordRGeodesicCover} (hereafter, \textit{\wordGRCover}) of $P$ if (i) each $Q \in \Q$ is a \GRpiece, and (ii) $\bigcup_{Q \in \Q} Q = P$. In particular, note that each $Q \in \Q$ is a subset of $P$. A \wordGRCover $\Q$ is a \textit{\wordGRPartition} if, for any distinct $Q, Q' \in \Q$, $Q$ and $Q'$ are interior disjoint. A \textit{\GeodesicDpiece} (hereafter, \textit{\GDpiece}), \textit{\wordGDCover}, and \textit{\wordGDPartition} are defined similarly.

We are now ready to define our problems of interest. Given a polygonal domain $P$ and an integer $k \in \naturals$, \GRCover (resp.\ \GRPart) asks whether there exists a \wordGRCover (resp.\ \wordGRPartition) $\Q$ of $P$ with $|\Q| \leq k$. \GDCover and \GDPart are similarly defined in terms of \GDpiece{}s.

It is known that \GRCover is \nph in polygonal domains~\cite{vigan_packing_2013}. Their construction can be readily adapted to show that the other three problems are also \nph.

\begin{theorem}[Theorem 10 in \cite{vigan_packing_2013}]
    \GRCover, \GRPart, \GDCover, and \GDPart are \nph in polygonal domains.
\end{theorem}
While it is not known whether these problems remain \nph in simple polygons, there has been extensive work dating back nearly a century within the mathematics community on determining the minimum number of unit balls required to cover constant-complexity regions (such as squares, triangles or disks of larger radius). Despite this long history, even some very basic questions in this area remain unresolved; see, for example, \cite{kershner_number_1939,verblunsky_least_1949,weisstein_disk_nodate,friedman_circles_nodate,friedman_circles_nodate-1,abrahamsen_partitioning_2022} and references therein. Recently, Abrahamsen and Stade showed that covering and partitioning polygons with pieces that fit inside a unit square is \nph \cite{abrahamsen_hardness_2024}.   

\subsection{Overview of results}
As our first result, we establish in Section~\ref{sec:equivalence} that \GRCover and \GRPart are equivalent in the sense that they admit optimal solutions of the same size. 

\begin{theorem}\label{thm:coverispartition}
    Let $\Q$ and $\Q'$ denote an optimal \wordGRCover and \wordGRPartition of a polygonal domain $P$ respectively. Then, $|\Q| = |\Q'|$.
\end{theorem}
Theorem~\ref{thm:coverispartition} (along with the fact that we can \textit{compute} a \wordGRPartition from a \wordGRCover efficiently) allows us to focus on coverings without loss of generality. Building on this, we design (in Section~\ref{sec:radius}) the first constant-factor approximation algorithms for \GRCover and \GRPart, thereby resolving an open question of Abrahamsen and Rasmussen~\cite{abrahamsen_partitioning_2022}. We use $\ordertilde{\cdot}$ to hide logarithmic factors, $n$ to denote the number of vertices of $P$, and \opt for the optimal number of small pieces in the problem of interest.

\begin{restatable}{theorem}{grapprox}\label{thm:GR_approx}
    Consider a simple polygon $P$ with $n$ vertices. There exists an algorithm running in $\ordertilde{\opt^3 \cdot n^2}$ time that returns a \wordGRCover of $P$ with at most $9 \cdot \opt$ pieces. Moreover, there is an algorithm which returns a \wordGRPartition of $P$ with the same running time and approximation guarantee.
\end{restatable}

Along the way, we also obtain a significantly simpler proof that the greedy boundary cover of Rabanca and Vigan~\cite[Theorem 1]{rabanca_covering_2015} yields a $2$-approximation. Our ideas are general enough to obtain a substantially improved approximation for \GDPart, achieving a $15$-approximation within the same running time. This not only improves the previously best-known factor of $72$~\cite[Table 1]{abrahamsen_partitioning_2022}, but also leads to a conceptually simpler algorithm and analysis. We discuss this in Section~\ref{sec:diameter}.

\begin{restatable}{theorem}{gdapprox}\label{thm:GD_approx}
    Consider a simple polygon $P$ with $n$ vertices. There is an algorithm running in $\ordertilde{\opt^3 \cdot n^2}$ time that returns a \wordGDPartition of $P$ with at most $15 \cdot \opt$ pieces.
\end{restatable}
We conclude the paper in Section~\ref{sec:conclusion} by discussing avenues for further research. 

\begin{remark}
Our techniques also apply to the setting where we measure geodesic distances by the $\ell_1$-metric, and can achieve matching approximations using a very similar argument. While there is not the same uncertainty of how to optimally cover even constant complexity convex regions with $\ell_1$ disks (squares) as there is with $\ell_2$ disks, the proof for NP-hardness for covering a \textit{simple} polygon with pieces that fit in unit squares from \cite{abrahamsen_hardness_2024} also applies to pieces with unit geodesic radius in the $\ell_1$ metric by rotating the construction by 45 degrees.
\end{remark}

\section{Equivalence of Small GR-Cover and Small GR-Partition}\label{sec:equivalence}
We begin this section by proving our first (and simplest) result, Theorem~\ref{thm:coverispartition}. To this end, we first establish a useful lemma that will be used in its proof.

\begin{lemma}\label{lma:covertopartition}
    Consider a \wordGRCover $\Q$ of a polygonal domain $P$. There is a \wordGRPartition $\Q'$ of $P$ with $|\Q'| = |\Q|$.
\end{lemma}
\begin{proof}
    Consider the centers $C_\Q$ of the geodesic disks used by $\Q$ to cover $P$. 
    Construct the geodesic Voronoi diagram on $C_\Q$; let this be $\Q'$. 
    Every point $p$ in $P$ is covered by at least one geodesic disk in $\Q$, so the closest disk center $c \in C_\Q$ to $p$ has $\dg{c}{p} \leq 1$. 
    Thus, each Voronoi cell is a \GRpiece showing that $\Q'$ is a \wordGRPartition as required.
\end{proof}

\begin{proof}[Proof of Theorem~\ref{thm:coverispartition}]
    Consider an optimal \wordGRCover $\Q$ and an optimal \wordGRPartition $\Q'$ of a polygonal domain $P$. Since all \wordGRPartition{}s are \wordGRCover{}s, $|\Q| \leq |\Q'|$. By Lemma~\ref{lma:covertopartition}, there is a \wordGRPartition $\Q''$ of $P$ with $|\Q''| = |\Q|$. Therefore, $|\Q'| \leq |\Q''| = |\Q|$; showing that $|\Q| = |\Q'|$ as required.
\end{proof}
We note that a similar result was established in \cite[Theorem~10]{aamand_covering_2026} for the case of covering and partitioning polygonal domains with pieces that fit inside a unit square. 
We conclude this section with a useful remark.

\begin{remark}\label{rmk:partition_from_centers}
    The geodesic Voronoi diagram of a simple polygon $P$ for a set of $m$ centers can be computed in $\order{n + m \log{m}}$ time \cite[Theorem 8.1]{oh_optimal_2019}.
    Since a \wordGRCover $\Q$ of $P$ can be efficiently represented by its centers $C_\Q$, we can compute a \wordGRPartition $\Q'$ of $P$ with $|\Q'| = |\Q|$ from (the centers $C_\Q$ of) $\Q$ in $\order{n + |\Q| \log{|\Q|}}$ time.
\end{remark}

\section{Approximation Algorithm for Small GR-Cover and Partition}\label{sec:radius}

In this section, we prove Theorem~\ref{thm:GR_approx}. Throughout this section, \opt denotes the number of pieces in an optimal \wordGRCover of the input (simple) polygon $P$. We first cover the boundary of $P$ using at most $2 \cdot \opt$ pieces (we actually only define a set of centers; the pieces are simply the unit geodesic ball centered on this set), using an algorithm by Rabanca and Vigan~\cite{rabanca_covering_2015}. Their correctness proof is quite involved; we simplify it substantially by invoking a result of Abrahamsen and Rasmussen~\cite{abrahamsen_partitioning_2022} --- see Lemma~\ref{lma:boundary_partition_greedy}. 
We then place an additional $2 \cdot \opt$ centers to cover points ``close'' to $\partial P$. 
Although this redundancy may appear wasteful, it provides a crucial structural property (Lemma~\ref{lma:covers_half}), which enables a simple greedy strategy for covering the interior using at most $5 \cdot \opt$ pieces. 
In Section~\ref{sec:Gr_runtime}, we show that all of these subroutines can be implemented in polynomial time, thereby establishing Theorem~\ref{thm:GR_approx}.

\subsection{Algorithm description}\label{sec:GR_algo}
We begin by stating a lemma concerning a partition of the boundary of a simple polygon, which is a slight generalization of a result by Abrahamsen and Rasmussen~\cite[Lemma~3]{abrahamsen_partitioning_2022}. Their proof immediately extends to our setting:

\begin{lemma}[Lemma 3 of \cite{abrahamsen_partitioning_2022}]\label{lma:boundary_partition_nongreedy}
    Let $\S$ be a partition of a simple polygon $P$ into connected pieces and let $\S_b = \{S \in \S \mid \partial S \cap \partial P \neq \emptyset\}$, i.e., $\S_b$ is the set of pieces in $\S$ which touch the boundary of $P$. Then, there is a partition of $\partial P$ into contiguous intervals $\I$ such that each interval $I \in \I$ is a subset of some $S \in \S_b$ and $|\I| \leq \max{(1, 2 \cdot |\S_b|-2)}$. 
\end{lemma}
We describe the algorithm \GRgreedyboundaryalgo, given by Rabanca and Vigan in \cite[Section 2.2]{rabanca_covering_2015}, to cover the boundary of $P$ with \GRpiece{}s. 
Start from a vertex $v$ of $P$ and place disk centers greedily: beginning at $v$, cover as far along $\partial P$ (walking counterclockwise) as possible such that the traversed contiguous interval $I$ is contained in a unit geodesic disk. 
Place a center to cover $I$ and repeat until returning to $v$. Let $\opt'$ denote the optimal number of \GRpiece{}s required to cover $\partial P$ --- note that $\opt' \leq \opt$.
Moreover, from (the proof of) Theorem~\ref{thm:coverispartition}, $\opt'$ is also the optimal number of \GRpiece{}s required to \textit{partition} $\partial P$.
We now present a significantly simpler proof than that of \cite{rabanca_covering_2015} of the following lemma:

\begin{lemma}[Theorem 1 of \cite{rabanca_covering_2015}]\label{lma:boundary_partition_greedy}
    \GRgreedyboundaryalgo places at most $2 \cdot \opt' - 1$ disk centers.
\end{lemma}

Consider an optimal partition $\S_b$ of $\partial P$ with \GRpiece{}s; $|\S_b| = \opt'$.
Let $\S_i$ denote the collection of connected pieces of the region inside $P$ that $\S_b$ leaves uncovered.
Let $\S = \S_b \cup \S_i$.
By Lemma~\ref{lma:boundary_partition_nongreedy}, there is a partition of $\partial P$ into contiguous intervals $\I$ where each $I \in \I$ is a \GRpiece and $|\I| \leq 2 \cdot |\S_b| - 2 \leq 2 \cdot \opt' - 2$.
Let $\I' = \{I'_1, I'_2 \dots I'_{k}\}$ denote the contiguous intervals constructed by \GRgreedyboundaryalgo{}.
We will show that $|\I'|= k \leq |\I| + 1 \leq 2\cdot \opt' -1$.
Let $S' = \{s'_i \mid 1 \leq i \leq k\}$ where $s'_i$ is the start-point of $I'_i \in \I'$; consider an $s'_j \in S'$ where $j \leq k - 1$.
For two points $a,b \in \partial P$, we use $\partial P [a,b]$ to denote the counterclockwise walk along $\partial P$ from $a$ to $b$.
Since $\I$ partitions $\partial P$, there is an $I = \partial P[s_I,t_I] \in \I$ which contains $s'_j$.
Recall that $I$ is a \GRpiece; if $s'_{j+1} \in \partial P[s_I,t_I)$, then $I'_j = \partial P[s'_j, s'_{j+1}]$ would not be maximal walk, a contradiction. So:
\begin{observation}\label{obs:next_one_is_different}
    $s'_{j+1} \notin \partial P[s_I,t_I)$.
\end{observation}
\begin{figure}[ht]
    \centering
    \includegraphics[width=0.7\linewidth,page=8]{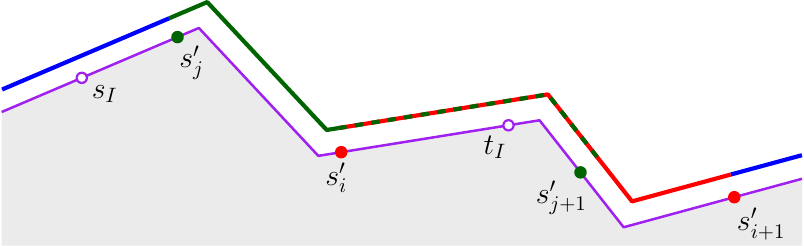}
    \caption{Illustration of Corollary~\ref{coro:different_intervals}. The boundary of $P$ is marked purple; the gray region is outside $P$. If $s'_i,s'_j \in \partial P[s_I,t_I)$, then $s'_{j+1}$ lies after $s'_i$. Therefore $I'_i$ (marked red above $\partial P$ for clarity) and $I'_j$ (marked green) overlap and therefore $\{I'_{\ell} \mid i \leq \ell \leq j\}$ covers $\partial P$.}
    \label{fig:wraparound}
\end{figure}
\begin{corollary}\label{coro:different_intervals}
    No two points in $S' \setminus \{s'_k\}$ are in $\partial P[s_I,t_I)$, for any $I \in \I$.
\end{corollary}
\begin{proof}
    Assume $s'_i$ and $s'_j$, with $1 \leq i < j \leq k - 1$, are both in $\partial P[s_I,t_I)$, for some $I \in \I$. 
    By Observation~\ref{obs:next_one_is_different}, $j \neq i + 1$; moreover, $s'_{j+1} \not\in \partial P[s_I,t_I)$ and therefore lies after $s'_i$ along $\partial P$.
    Therefore, $\{I'_{\ell} \mid i \leq \ell \leq j\}$, a strict subset of $\I'$, covers $\partial P$, a contradiction. See Figure~\ref{fig:wraparound}.
\end{proof}
We can now prove Lemma~\ref{lma:boundary_partition_greedy}.
\begin{proof}[Proof of Lemma~\ref{lma:boundary_partition_greedy}]
    By Corollary~\ref{coro:different_intervals}, at most one point in $S' \setminus \{s'_k\}$ is in each interval of $\I$. 
    Therefore, $k-1 \leq |\I| \leq 2 \cdot \opt' - 2$; so $|\I'| = k \leq 2 \cdot \opt' - 1$ as required.
\end{proof}

Let $C'_b$ denote the centers placed by \GRgreedyboundaryalgo. We add additional (and, at first glance, redundant) centers to $C'_b$. Let $\I$ denote the collection of contiguous intervals considered by \GRgreedyboundaryalgo{} --- clearly, $|C'_b| = |\I|$. At every endpoint of an interval of $\I$, we place a new geodesic disk center and add this to $C'_b$. We call this subroutine \GRextracentersalgo.

\begin{algorithm}[ht]
\caption{\GRextracentersalgo}\label{alg:extra_centers}
\begin{algorithmic}[1]
\State \textbf{Input:} A set of centers $C'_b$ and a partition $\I$ of $\partial P$.
\State \textbf{Output:} A set of centers $C_b$. 

\State $C_b \gets C'_b$

\For{$I$ in $\I$}
    \State Append $I[0]$ to $C_b$
\EndFor

\State \Return $C_b$

\end{algorithmic}
\end{algorithm}

Let $C_b$ denote the resulting set of centers. 
Clearly, $|C_b| = 2 \cdot|\I|$. 
Even though this is more centers than necessary to cover just the boundary of $P$, this cover gives some useful properties for constructing a cover of the interior. 
Let $\Q_b$ denote the collection of unit geodesic discs centered at $C_b$ --- i.e., $\Q_b = \{\ballg{c} \mid c \in C_b\}$.
Since $|\Q_b| = |C_b| = 2 \cdot |\I|$ and $|\I| \leq 2 \cdot \opt' \leq 2 \cdot \opt$ (by Lemma~\ref{lma:boundary_partition_greedy}), we have:

\begin{corollary}\label{cor:boundary_pieces}
    $|\Q_b| \leq 4 \cdot \opt$. 
\end{corollary}
We now show that $\Q_b$ not only covers $\partial P$ but also all points that are ``close'' (in the Euclidean sense) to $\partial P$.
Recall that $\spath{p}{q}$ denotes the shortest path between $p$ and $q$ within $P$. 
\begin{lemma}\label{lma:covers_half}
    $\Q_b$ covers every point of $P$ that is a \textbf{Euclidean} distance at most $\frac{1}{2}$ from $\partial P$. 
\end{lemma}
\begin{proof}
    Consider a point $p \in P$ which is within a distance of $\frac{1}{2}$ from $\partial P$. Consider $q \in \partial P$, a point on the boundary closest to $p$. We have $\de{p}{q} \leq \frac{1}{2}$ and that no point in the interior of $\overline{pq}$ is on $\partial P$ --- otherwise there is another point $q'$ closer to $p$ on $\partial P$. So, we have $\overline{pq} \subset P$. Since $\I$ covers $\partial P$, $q$ lies in some contiguous interval $I \in \I$ of $\partial P$. Let $I = \partial P[s_I, t_I]$ and let $c_I \in C'_b \subseteq C_b$ be the geodesic disk center placed by \GRgreedyboundaryalgo for $I$. We have two cases (see Figure~\ref{fig:euclidean_at_most_half}): $\overline{pq}$ intersects $\spath{s_I}{c_I} \cup \spath{c_I}{t_I}$ or it does not.
    \begin{figure}[ht]
        \centering
        \includegraphics[width=0.5\linewidth,page=5]{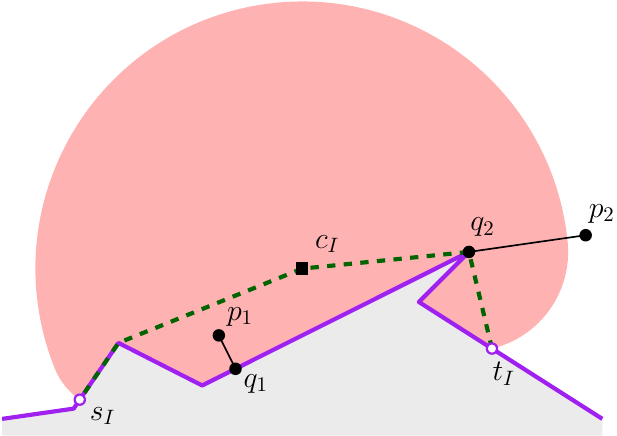}
        \caption{Illustration of Lemma~\ref{lma:covers_half}. The boundary of $P$ is drawn in purple, the gray region is outside $P$. The shortest paths from $c_I$ to $s_I$ and $t_I$ are in dashed green. We have $p_1,p_2 \in P$ and they are at most an Euclidean distance of $\frac{1}{2}$ away from $\partial P$. The points $q_1,q_2 \in \partial P$ are the closest point on $\partial P$ to $p_1$ and $p_2$. Since $\overline{p_1q_1} \cap (\spath{s_I}{c_I} \cup \spath{c_I}{t_I}) = \emptyset$, $p_1 \in \ballg{c_I}$. On the other hand, $\overline{p_2q_2} \cap \spath{c_I}{t_I} \neq \emptyset$; since this intersection point is closer to $t_I$ than to $c_I$, $p_2 \in \ballg{t_I}$.}
        \label{fig:euclidean_at_most_half}
    \end{figure}
    
    \begin{case}Assume that $\overline{pq}$ does not intersect $\spath{s_I}{c_I} \cup \spath{c_I}{t_I}$. Then, $p$ must lie inside the region\footnote{$R$ is a degenerate weakly simple polygon as some of its sides could overlap --- see Figure~\ref{fig:euclidean_at_most_half}.} $R$ within the closed curve $\spath{s_I}{c_I} \cup \spath{c_I}{t_I} \cup \partial P[t_I,s_I]$ by the Jordan curve theorem. Since $I$ is contiguous, $R$ is within the geodesic disk centered at $c_I$, and thus $p$ is covered by $\Q_b$.
    \end{case}
    \begin{case}
        Assume that $\overline{pq}$ intersects $\spath{s_I}{c_I} \cup \spath{c_I}{t_I}$. We assume that there is a point $q' \in \overline{pq} \cap \spath{c_I}{t_I}$; the other case is symmetric. Note that $\dg{c_I}{t_I} \leq 1$ --- i.e., the length of $\spath{c_I}{t_I}$ is at most 1. Since $q' \in \spath{c_I}{t_I}$, either $\dg{c_I}{q'} \leq \frac{1}{2}$ or $\dg{t_I}{q'} \leq \frac{1}{2}$. Since $\de{q'}{p} \leq \de{q}{p} \leq \frac{1}{2}$, $p$ is within a geodesic distance of $1$ from either $c_I$ or $t_I$. Since both $c_I$ and $t_I$ belong to $C_b$, it follows that $p$ is covered by $\Q_b$.
    \end{case}
    These cases are exhaustive, so the claim follows.
\end{proof}
The above property implies the following corollary:

\begin{corollary}\label{cor:euclidean_reduction}
    Consider two points $p$ and $q$ in $P$ which are not covered by $\Q_b$. If $\de{p}{q} \leq 1$, $\dg{p}{q} \leq 1$.
\end{corollary}
\begin{proof}
    If $\overline{pq} \cap \partial P = \emptyset$, then $\spath{p}{q} = \overline{pq}$; so $\dg{p}{q} = \de{p}{q} \leq 1$ as required. Otherwise, there is a $r \in \overline{pq} \cap \partial P$; so $\de{r}{p}$ or $\de{r}{q}$ is at most $\frac{1}{2}$ (since $|\overline{pq}|\leq 1$), which, by Lemma~\ref{lma:covers_half}, is a contradiction to the assumption that $p$ and $q$ are uncovered by $\Q_b$.
\end{proof}

With this in mind, we proceed to cover the interior of $P$ with a greedy algorithm which we call \greedyinterioralgo. Consider the uncovered region $P' = P \setminus \cup_{Q \in \Q_b}Q$ and place a center at a point $c$ in the interior of $P'$. Remove $\ballg{c}$ from $P'$ and iterate until all points in $P'$ have been covered. Let $C_i$ denote the centers placed by \greedyinterioralgo in $P'$. Define $\Q_i = \{\ballg{c} \mid c \in C_i\}$ and let $\Q = \Q_b \cup \Q_i$.

\begin{algorithm}[ht]
\caption{\greedyinterioralgo}\label{alg:greedy_interior}
\begin{algorithmic}[1]
\State \textbf{Input:} A simple polygon $P$, an uncovered region $U$, and a radius $r$.
\State \textbf{Output:} A set of centers $C_i$. 

\State $C_i \gets \emptyset$
\State $R \gets U$

\While{$R \neq \emptyset$}\label{line:check_emptyset}
    \State Find a point $c$ in the interior of $R$ \label{line:find_interior_point}
    \State Append $c$ to $C_i$
    \State $R \gets R \setminus \ballgr[P]{r}{c}$ \label{line:maintain_overlay}
\EndWhile
\State \Return $C_i$
\end{algorithmic}
\end{algorithm}

\begin{observation}\label{obs:centers_are_far}
    For any two distinct centers $c$ and $c'$ in $C_i$, $\de{c}{c'} > 1$.
\end{observation}
\begin{proof}
    Since \greedyinterioralgo picks centers iteratively from the interior of the uncovered region, we have that $\dg{c}{c'} > 1$ for any two centers $c$ and $c'$ in $C_i$. The claim now follows from (the contrapositive of) Corollary~\ref{cor:euclidean_reduction}.
\end{proof}
\begin{observation}\label{obs:covers_poly}
     $\Q$ is a \wordGRCover of $P$.
\end{observation}
\begin{proof}
    Consider a point $p \in P$ which is not covered by $\Q_b$ --- i.e, $p \in P'$. Since \greedyinterioralgo iterates until all points in $P'$ are covered, $p$ is covered by $\Q'_i$.
\end{proof}
To achieve the approximation ratio promised in Theorem~\ref{thm:GR_approx}, we prove that $|\Q| \leq 9 \cdot \opt$. From Corollary~\ref{cor:boundary_pieces}, we already have $|\Q_b| \leq 4 \cdot \opt$, so it remains to show the following:

\begin{lemma}\label{lma:interior_approx}
    $|\Q_i| \leq 5 \cdot \opt$.
\end{lemma}
We first establish two results about Euclidean disks that we will use to prove the above lemma. For a point $c \in \reals^2$ and angles $0 \leq \alpha < \beta \leq 2\pi$, let $\sector{\alpha}{\beta}$ denote the \textit{sector} of $\balle{c}$ induced by the two radii that make angles $\alpha$ and $\beta$ with the $x$--axis. Note that a sector is a closed set; in particular, it contains the two defining radii. A \textit{sextant} $\sextant{\alpha}$ is the sector $\sector{\alpha}{\alpha + \frac{\pi}{3}}$. See Figure~\ref{fig:sector}, \textit{Left}.
\begin{figure}[ht]
    \centering
    \includegraphics[width=0.95\linewidth,page=6]{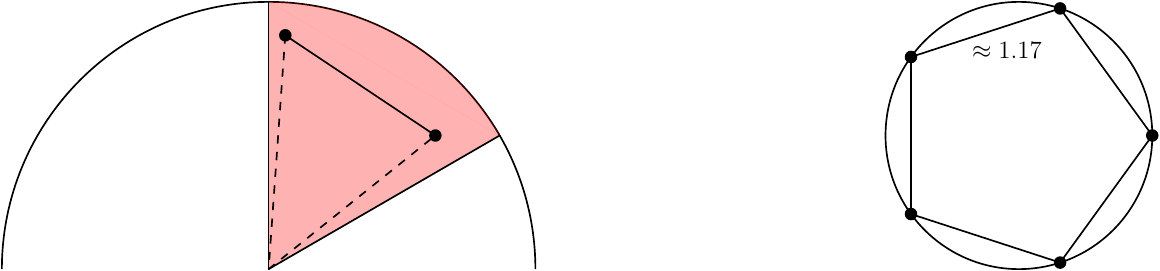}
    \caption{Left: The sextant $\sextant{\frac{\pi}{6}}$ is in red. Two points inside a sextant are at most a distance 1 apart since they subtend an angle of at most $\frac{\pi}{3}$ between them at the center. Right: The distance between 5 equally spaced points along a unit circle is at least $1.17$.}
    \label{fig:sector}
\end{figure}

\begin{observation}\label{obs:sector_contains_1}
    For any two points $p$ and $q$ in an arbitrary sextant $\sextant{\alpha}$, $\de{p}{q} \leq 1$. 
\end{observation}
\begin{proof}
    Since $\de{c}{p}$ and $\de{c}{q}$ are at most 1 and the angle between $\overline{cp}$ and $\overline{cq}$ is at most $\frac{\pi}{3}$, our claim follows. See Figure~\ref{fig:sector}, \textit{Left}.
\end{proof}
\begin{corollary}\label{coro:ball_contains_5}
    There \textbf{does not} exist a subset $X \subseteq \balle{c}$ with $|X| \geq 6$ such that $\de{p}{q} > 1$ for any two points $p$ and $q$ in $X$. 
\end{corollary}
\begin{proof}
    Assume that there exists such a set $X$ with $|X| \geq 6$. Then, there exists two points $p,q \in X$ which subtend an angle of at most $\frac{\pi}{3}$ between them at $c$. This implies that $\de{p}{q} \leq 1$ from Observation~\ref{obs:sector_contains_1}, a contradiction. 
\end{proof}
Note that Corollary~\ref{coro:ball_contains_5} is tight --- Let $X$ denote the set of 5 equally spaced points along the unit circle; the distance between any two points in $X$ is at least $1.17$. See Figure~\ref{fig:sector}, \textit{Right}.
We are now ready to prove Lemma~\ref{lma:interior_approx}.
\begin{proof}[Proof of Lemma~\ref{lma:interior_approx}]
    Consider an optimal \wordGRCover $\Q^*$ of $P$ and a \GRpiece $Q \in \Q^*$. Let $c$ be the center of $Q$; we have $Q = \ballg{c} \subseteq \balle{c}$. Let $X = C_i \cap \balle{c}$ --- Note that $X \supseteq C_i \cap Q$. Moreover, every distinct pair of points in $X$ are more than a distance 1 apart (Observation~\ref{obs:centers_are_far}). From Corollary~\ref{coro:ball_contains_5}, we have that $|X| \leq 5$. Since every point in $C_i$ is covered by some \GRpiece in $\Q^*$, we have $|C_i| \leq 5 \cdot |\Q^*|$; the result now follows.    
\end{proof}
\subsection{Runtime analysis}\label{sec:Gr_runtime}

We now show that the three subroutines used by our algorithm \GRgreedypolygonalgo to construct $\Q$ --- namely, \GRgreedyboundaryalgo, \GRextracentersalgo, and \greedyinterioralgo{} --- can be executed in polynomial time.

\begin{algorithm}[ht]
\caption{\GRgreedypolygonalgo}\label{alg:greedy_full}
\begin{algorithmic}[1]
\State \textbf{Input:} A simple polygon $P$ stored as a circular list.
\State \textbf{Output:} A \wordGRCover $\Q$ of $P$. 

\State $(C'_b, \I) \gets \GRgreedyboundaryalgo(P)$\label{line:run_boundary}
\State $C_b \gets \GRextracentersalgo(C'_b, \I)$\label{line:run_extra_centers}
\State $\Q_b \gets \{\ballg{c} \mid c \in C_b\}$ \label{line:compute_Qb}
\State $P' \gets P \setminus \cup_{Q \in \Q_b}Q$ \label{line:compute_P'}
\State $C_i \gets \greedyinterioralgo(P, P',1)$
\State $\Q_i \gets \{\ballg{c} \mid c \in C_i\}$\label{line:interior_pieces}
\State $\Q = \Q_b \cup \Q_i$\label{line:total_pieces}

\State \Return $\Q$
\end{algorithmic}
\end{algorithm}
The first two steps of \GRgreedypolygonalgo is to run \GRgreedyboundaryalgo and \GRextracentersalgo (Lines~\ref{line:run_boundary} and \ref{line:run_extra_centers}). 
\GRgreedyboundaryalgo was shown to run in polynomial time by Rabanca and Vigan in \cite{rabanca_covering_2015}. 
%A slight tweak in \GRgreedyboundaryalgo leads to the faster runtime given in \cite{rabanca_covering_2015}, but we state the simpler version here since the two remaining subroutines dominate the overall runtime anyway.

\begin{lemma}[Section 2.2 in \cite{rabanca_covering_2015}]\label{lma:constructing_boundary_pieces}
    Given a polygon $P$, \GRgreedyboundaryalgo returns $C'_b$ and $\I$ in $\ordertilde{|C'_b| + n}$ time.
\end{lemma}
%We note that Lines~\ref{line:finding maximal walk} and~\ref{line:finding next vertex} in \GRgreedyboundaryalgo require some care in both their implementation and analysis. 
%In contrast, the running time of \GRextracentersalgo is immediate.
The running time of \GRextracentersalgo is immediate.
\begin{lemma}\label{lma:constructing_extra_centers}
    Given a set of centers $C'_b$ and a patition $\I$ of $\partial P$, \GRextracentersalgo returns $C_b$ in $\order{|\I| + |C_b|}$ time.
\end{lemma}
Next, we compute $P'$ using $\Q'_b$ (Lines~\ref{line:compute_Qb} and \ref{line:compute_P'} in \GRgreedypolygonalgo). $\Q_b$ can be constructed in $\order{|C_b| \cdot n^2}$ using an algorithm by Borgelt, van Krevald, and Luo \cite{borgelt_geodesic_2011}. Since a geodesic ball is described by $\order{n}$ circular arcs and line segments, the arrangement of $P$ and $\Q_b$ has complexity $\order{|C_b|^2 \cdot n^2}$; and can be computed in that time \cite{halperin_arrangements_2004}. Thus, Line~\ref{line:compute_P'} takes $\order{|C_b|^2 \cdot n^2}$ time.

\begin{lemma}\label{lma:constructing_interior_pieces}
    Given an uncovered region $P' \subset P$ (as an overlay of $|C_b|$ geodesic disks and $P$), \greedyinterioralgo returns $C_i$ in $\ordertilde{|C_i| \cdot (|C_b|+|C_i|)^2 \cdot n^2}$ time. 
\end{lemma}
\begin{proof}
    We can maintain $R$, the uncovered region through the course of the algorithm (Line~\ref{line:maintain_overlay} in \greedyinterioralgo) in $\ordertilde{(|C_b|+|C_i|)^2 \cdot n^2}$ time using an incremental construction of the arrangement. See, for example, \cite[Theorem 28.4.2]{halperin_arrangements_2004} and \cite[Section 6]{sack_chapter_2000}. Maintaining this arrangement also allows for us to check if $R$ is empty (Line~\ref{line:check_emptyset}) in constant time in each iteration of the loop. 

    To find a point $c$ in the interior of $R$ (Line~\ref{line:find_interior_point}), we find the minimum enclosing (Euclidean) ball $B$ of the vertices of $R$. Let $v$ be one of the vertices of $R$ on $\partial B$ and let $r$ be the radius of $B$. Consider the ball $B'$ of radius $r$ around $v$; $\partial B'$ intersects the interior of $R$ in (possibly) disconnected arcs. Let $c$ be the midpoint of one such arc\footnote{To bypass bit complexity issues, we assume the standard Real-RAM model.}. Finding $B$ takes linear (in the complexity of $R$) time; see, for example, \cite[Chapter 4.7]{de_berg_computational_2008}. Since this complexity is bounded by $\order{(|C_b|+|C_i|)^2 \cdot n^2}$ and there are $|C_i|$ iterations of the loop, the runtime follows.  
\end{proof}
We now prove the main result of this section, Theorem~\ref{thm:GR_approx} (repeated here for readability):

\grapprox*
\begin{proof}
    Consider the \wordGRCover $\Q$ that \GRgreedypolygonalgo returns. From Corollary~\ref{cor:boundary_pieces} and Lemma~\ref{lma:interior_approx}, we have that $|\Q| \leq 9 \cdot \opt$. From Lemmas~\ref{lma:constructing_boundary_pieces}, \ref{lma:constructing_extra_centers} and \ref{lma:constructing_interior_pieces}, it is clear that its running time is dominated by the subroutine \greedyinterioralgo. The runtime guarantee follows since $|C_i| + |C_b| = |\Q| \leq 9 \cdot \opt$.

    To return a \wordGRPartition instead of a \wordGRCover, we modify Lines~\ref{line:interior_pieces} and \ref{line:total_pieces} of \GRgreedypolygonalgo slightly: Compute the geodesic Voronoi diagram of $P$ for the centers $C_b \cup C_i$ and return as $\Q$ the Voronoi cells instead (see Remark~\ref{rmk:partition_from_centers}). 
    By Theorem~\ref{thm:coverispartition}, we have that $|\Q| \leq 9 \cdot \opt$; the running time is unaffected since the geodesic Voronoi diagram can be computed in $\order{n + \opt \log{\opt}}$ time using \cite[Theorem 8.1]{oh_optimal_2019}.
\end{proof}

\section{Approximation Algorithm for Small GD-Partition}\label{sec:diameter}

Our goal in this section is to prove Theorem~\ref{thm:GD_approx}. Let $\opt$ denote the number of pieces in an optimal \wordGDPartition of the input simple polygon $P$. Structurally, the algorithm we describe is similar to that in Section~\ref{sec:radius}. We first partition the boundary $\partial P$ into $2 \cdot \opt$ pieces using an algorithm of Abrahamsen and Rasmussen~\cite{abrahamsen_partitioning_2022} (Lemma~\ref{lma:GD_boundary_cover}). We then place at most $4 \cdot \opt$ additional pieces to cover points ``close'' to $\partial P$, and finally cover the interior using \greedyinterioralgo.
In contrast to Section~\ref{sec:radius}, where the interior can be covered using at most $5 \cdot \opt$ pieces, we obtain a bound of $9 \cdot \opt$ pieces here (Lemma~\ref{lma:interior_approx_part}). We conclude in Section~\ref{sec:Gd_runtime} that the running time of the algorithm is once again dominated by \greedyinterioralgo, thereby establishing Theorem~\ref{thm:GD_approx}.

\subsection{Algorithm description}
We begin by describing \GDgreedyboundaryalgo, the algorithm designed by Abrahamsen and Rasmussen~\cite[Sections 3.1 and 3.2]{abrahamsen_partitioning_2022} to cover the boundary of a simple polygon $P$ with \GDpiece{}s.
It is nearly identical to \GRgreedyboundaryalgo that we saw in the previous section: Start from a vertex $v$ of $P$ and cover as far along $\partial P$ (walking counterclockwise) as possible such that the traversed contiguous interval $I$ has geodesic diameter at most 1. 
From $I = \partial P[s_I, t_I]$, construct $\spath{t_I}{s_I}$ to obtain the \GDpiece bounded by $\partial P[s_I, t_I]$ and $\spath{t_I}{s_I}$. Repeat until returning to $v$.

Let $\Q'_b$ and $\I$ denote the collection of \GDpiece{}s and the partition of $\partial P$ the above algorithm returns respectively.
\begin{lemma}[Lemma 4 of \cite{abrahamsen_partitioning_2022}]\label{lma:GD_boundary_cover}
    $|\Q'_b| = |\I| \leq 2 \cdot \opt - 1$.
\end{lemma}
Consider $P'$, the region of $P$ uncovered by $\Q'_b$. Note that $P'$ is a (collection of) degenerate weakly simple polygon(s). since $\I \in \order{\opt}$, it follows that $|P'| \in \order{\opt \cdot n}$.
\begin{corollary}\label{coro:perimeter_small}
    The perimeter of $P'$ is at most $2 \cdot \opt -1$.
\end{corollary}
\begin{proof}
    Consider a \GDpiece $Q \in \Q'_b$. Let $I = \partial P[s_I,t_I]$ be the interval on $\partial P$ that it corresponds to. 
    Observe that $\partial Q \cap \partial P' \subseteq \spath{s_I}{t_I}$ by construction. 
    Since $|\spath{s_I}{t_I}| \leq 1$ ($Q$ is a \GDpiece), we infer that $Q$  contributes a length of at most 1 to the perimeter of $P'$. 
    As $|\Q'_b| \leq 2 \cdot \opt - 1$ (see Lemma~\ref{lma:GD_boundary_cover}), the claim immediately follows.
\end{proof}
We will cover $P'$ using geodesic balls of \textit{radius} $\frac{1}{2}$ --- Note that these have geodesic \textit{diameter} at most 1 as required. 
To begin, we construct a cover $\Q''_b$ of the boundary of $P'$: Start with a vertex $v$ of $P'$ and place centers at every $\frac{1}{2}$ interval along $\partial P'$. 
We call this subroutine \GDextracentersalgo.
\begin{algorithm}[ht]
\caption{\GDextracentersalgo}\label{alg:GD_extra_centers}
\begin{algorithmic}[1]
\State \textbf{Input:} A region $P'$.
\State \textbf{Output:} A set of centers $C''_b$. 

\State $C''_b \gets \emptyset$
\State $s \gets v_0$ \Comment{$v_i$ denotes the vertex $P[i]$}
\State $i \gets 1$ \Comment{$i$ holds the index of the vertex right after $s$}
\State flag $\gets \false$ \Comment{Checks if we have walked through all of $\partial P$}

\While{flag is \false}
    \State Append $s$ to $C''_b$
    \State $j \gets$ the index of the first vertex such that $|\partial P'[s,v_j]| > \frac{1}{2}$
    \State $t \gets$ maximal point on $\overline{v_{j-1}v_{j}}$ such that $|\partial P'[s,t]| = \frac{1}{2}$
    \State $s \gets t$
    \State $i \gets j$
    \State If we have walked through all of $\partial P$, set flag $\gets \true$
\EndWhile

\State \Return $C''_b$

\end{algorithmic}
\end{algorithm}

Let $C''_b$ denote the centers placed by \GDextracentersalgo and define $\Q''_b = \{\ballgr[P']{\frac{1}{2}}{c} \mid c \in C''_b\}$. 
Corollary~\ref{cor:boundary_pieces_part} is a direct consequence of Corollary~\ref{coro:perimeter_small}.

\begin{corollary}\label{cor:boundary_pieces_part}
    $|\Q''_b| \leq 4 \cdot \opt$. 
\end{corollary}
We show a complementary lemma to Lemma~\ref{lma:covers_half}: 
$\Q_b:= \Q_b' \cup \Q''_b$ covers every point that is a Euclidean distance at most $\frac{1}{4}$ from $\partial P$. 
Then, Corollary~\ref{cor:GD_euclidean_reduction} and Observation~\ref{obs:gd_centers_are_far} follow from Lemma~\ref{lemma:covers_half_part} just as Corollary~\ref{cor:euclidean_reduction} and Observation~\ref{obs:centers_are_far} did from Lemma~\ref{lma:covers_half}.

\begin{lemma}\label{lemma:covers_half_part}
    $\Q_b$ covers every point of $P$ that is a \textbf{Euclidean} distance at most $\frac{1}{4}$ from $\partial P$.
\end{lemma}
\begin{proof}
    Consider a point $p \in P$ that is a Euclidean distance at most $\frac{1}{4}$ from $\partial P$.
    If $p \in P \setminus P'$, it must be covered by $\Q'_b$ by definition. 
    Assume that $p \in P'$. 
    Since $\de{p}{\partial P} \leq \frac{1}{4}$ and $\partial P'$ separates $\partial P$ and $P'$, $\de{p}{\partial P'} \leq \frac{1}{4}$. 
    Let $q \in \partial P'$ be the closest point to $p$ on $\partial P'$ and let $c$ be the closest center (where the distance is measured along $\partial P'$) in $C''_b$ to $q$. 
    Note that $\overline{pq} \subset P'$ --- otherwise there is a point on $\partial P'$ closer to $p$ than $q$.
    Therefore, we have that $\dg[P']{p}{q} = \de{p}{q} \leq \frac{1}{4}$. 
    Moreover, since \GDextracentersalgo places centers at every $\frac{1}{2}$ interval along $\partial P'$, $\dg[P']{c}{q} \leq \frac{1}{4}$.
    This shows that the path from $c$ to $p$ through $q$ has length at most $\frac{1}{2}$, and thus $p \in \ballgr[P']{\frac{1}{2}}{c}$.
    So $p$ is covered by $\Q''_b$ as required.
\end{proof}
\begin{corollary}\label{cor:GD_euclidean_reduction}
    Consider two points $p$ and $q$ in $P$ which are not covered by $\Q_b$. 
    If $\de{p}{q} \leq \frac{1}{2}$, $\dg{p}{q} \leq \frac{1}{2}$.
\end{corollary}
Let $P'' = P' \setminus \cup_{Q \in \Q''_b} Q$. 
We reuse \greedyinterioralgo to cover $P''$; we call \greedyinterioralgo with $U = P''$ and $r = \frac{1}{2}$ (Line~\ref{line:GD_interior}, \GDgreedypolygonalgo). 
$C_i$ denotes the centers placed by \greedyinterioralgo in $P''$; $\Q_i := \{\ballgr[P']{\frac{1}{2}}{c} \mid c \in C_i\}$. 
We have:
\begin{observation}\label{obs:gd_centers_are_far}
    For any two distinct centers $c$ and $c'$ in $C_i$, $\de{c}{c'} > \frac{1}{2}$.
\end{observation}
\begin{observation}\label{obs:GD_covers_P}
    $\Q_b \cup \Q_i$ covers every point in $P$.
\end{observation}
\begin{corollary}\label{cor:close_center}
    Consider a point $p \in P'$ and let $c \in C''_b \cup C_i$ be the center that is closest to $p$. Then, $\dg{c}{p} \leq \frac{1}{2}$.
\end{corollary}
\begin{proof}
    From Observation~\ref{obs:GD_covers_P}, we infer that $p$ is covered by a piece in $\Q''_b \cup \Q_i$.
    Since these are $\frac{1}{2}$-radius geodesic balls, our claim follows.
\end{proof}
Since we need to output a partition (and not just a cover), we compute the geodesic Voronoi diagram of $C_b'' \cup C_i$ inside $P'$ (see Remark~\ref{rmk:partition_from_centers}). 
Let $\Q$ denote the union of these Voronoi cells and $\Q'_b$.
\begin{lemma}\label{lma:is_GD_partition}
    $\Q$ is a \wordGDPartition of $P$.
\end{lemma}
\begin{proof}
    The pieces of $\Q'_b$ are interior disjoint from \cite[Sections 3.1 and 3.2]{abrahamsen_partitioning_2022}.
    Since the remaining pieces are from a Voronoi partition of $P \setminus \Q'_b$, it follows that $\Q$ is a partition of $P$.

    $\Q'_b$ contains \GDpiece{}s \cite[Sections 3.1 and 3.2]{abrahamsen_partitioning_2022}; consider a piece $Q \in \Q \setminus \Q'_b$.
    Then, $Q$ is a Voronoi cell of a center $c \in C''_b \cup C_i$.
    It follows from Corollary~\ref{cor:close_center} that $Q \subseteq \ballgr{\frac{1}{2}}{c}$.
    $Q$ is therefore a \GDpiece.
\end{proof}
From Lemma~\ref{lma:GD_boundary_cover} and Corollary~\ref{cor:boundary_pieces_part}, we have that $|\Q'_b| \leq 2 \cdot \opt$ and $|\Q''_b| \leq 4 \cdot \opt$ respectively; so to achieve the approximation guarantee promised in Theorem~\ref{thm:GD_approx}, we are left to show that $|\Q_i| \leq 9 \cdot \opt$. 
Note that we cannot use a similar argument as in Lemma~\ref{lma:interior_approx}: A \GDpiece does not necessarily need to fit inside a $\balle[\frac{1}{2}]{c}$ for some center $c \in \reals^2$; consider, for example, the equilateral triangle with side lengths 1.
See Section~\ref{sec:conclusion} for further discussion. We make two observations before proving that $|\Q_i| \leq 9 \cdot \opt$:
\begin{observation}\label{obs:fits_inside_square}
    Let $Q$ be a \GDpiece of $P$. Then, $Q$ fits inside a unit square.
\end{observation}
\begin{proof}
    Let $x_{\min}$, $x_{\max}$, $y_{\min}$, and $y_{\max}$ denote the smallest and largest $x$ and $y$-coordinates of points in $Q$. Since the diameter of $Q$ is at most 1, $x_{\max} - x_{\min} \leq 1$ and $y_{\max} - y_{\min} \leq 1$. Our claim now follows.
\end{proof}
\begin{observation}\label{obs:nine_inside_square}
    Consider a unit square $S$. There \textbf{does not} exist a subset $X \subseteq S$ with $|X| \geq 10$ such that $\de{p}{q} > \frac{1}{2}$ for any two points $p$ and $q$ in $X$. 
\end{observation}
\begin{proof}
    We partition $S$ into 9 equal-sized subsquares $\S = \{S_1, S_2 \dots S_9\}$. Note that the Euclidean diameter of a $S \in \S$ is $\frac{\sqrt{2}}{3} < \frac{1}{2}$. Therefore, each $S \in \S$ contains at most one point from $X$; this proves our claim.  
\end{proof}
\begin{lemma}\label{lma:interior_approx_part}
    $|\Q_i| \leq 9 \cdot \opt$.
\end{lemma}
\begin{proof}
   Consider an optimal \wordGDPartition $\Q^*$ of $P$ and a \GDpiece $Q \in \Q^*$. Then, there is a unit square $S$ such that $Q \subseteq S$ (see Observation~\ref{obs:fits_inside_square}). Let $X = C_i \cap S$; note that $X \supseteq C_i \cap Q$. Every distinct pair of points in $X$ are a Euclidean distance $\frac{1}{2}$ apart (from Observation~\ref{obs:gd_centers_are_far}). So, from Observation~\ref{obs:nine_inside_square}, we have $|X| \leq 9$. Since every point in $C_i$ is covered by some \GDpiece in $\Q^*$, we have $|C_i| \leq 9 \cdot |\Q^*|$; the result now follows. 
\end{proof}

\subsection{Runtime analysis}\label{sec:Gd_runtime}
We refer to our algorithm to find a \wordGDPartition of a simple polygon as \GDgreedypolygonalgo; we now prove that it runs in polynomial time. 

\begin{algorithm}[ht]
\caption{\GDgreedypolygonalgo}\label{alg:GD_greedy_full}
\begin{algorithmic}[1]
\State \textbf{Input:} A simple polygon $P$ stored as a circular list.
\State \textbf{Output:} A \wordGDPartition $\Q$ of $P$. 

\State $(\Q'_b, \I) \gets \GDgreedyboundaryalgo(P)$\label{line:GD_run_boundary}
\State $P' \gets P \setminus \cup_{Q \in Q'_b} Q$ \label{line:GD_compute_P'}
\State $C''_b \gets \GDextracentersalgo(P')$\label{line:GD_run_extra_centers}
\State $\Q''_b \gets \{\ballgr[P']{\frac{1}{2}}{c} \mid c \in C''_b\}$ \label{line:GD_compute_Qb}
\State $P'' \gets P' \setminus \cup_{Q \in \Q''_b}Q$ \label{line:GD_compute_P''}
\State $C_i \gets \greedyinterioralgo(P, P'', \frac{1}{2})$ \label{line:GD_interior}
\State $\Q \gets \Q'_b \cup \textsc{ComputeGeodesicVoronoi}(P', C''_b \cup C_i)$ \label{line:GD_total_pieces}
\State \Return $\Q$
\end{algorithmic}
\end{algorithm}

\begin{lemma}[Lemma 5 of \cite{abrahamsen_partitioning_2022}]\label{lma:GD_greedy_runtime}
    Given a polygon $P$, \GDgreedyboundaryalgo returns $\Q'_b$ and $\I$ in $\order{n^2 \log{n} + |\Q_b|}$ time.
\end{lemma}

Abrahamsen and Rasmussen \cite{abrahamsen_partitioning_2022} also construct $P'$ and the above runtime includes this construction of $P'$.
Lines~\ref{line:GD_compute_Qb} and \ref{line:GD_compute_P''} of \GDgreedypolygonalgo take $\order{|C''_b|^2 \cdot n^2}$ time, see the discussion between Lemmas~\ref{lma:constructing_extra_centers} and \ref{lma:constructing_interior_pieces} in Section~\ref{sec:Gr_runtime}.
Lemma~\ref{lma:constructing_interior_pieces} proves that \greedyinterioralgo takes $\ordertilde{|C_i| \cdot (|C''_b|+|C_i|)^2 \cdot n^2}$ time while Remark~\ref{rmk:partition_from_centers} shows that Line~\ref{line:GD_total_pieces} takes $\ordertilde{|P'| + |C''_b \cup C_i|}$ time.  
We ascertain the runtime of \GDextracentersalgo below and then prove Theorem~\ref{thm:GD_approx}, restated here for convenience.

\begin{observation}\label{obs:silly_boundary}
    Given a region $P'$, \GDextracentersalgo returns $C''_b$ in $O(|P'| + |C''_b|)$ time.
\end{observation}

\gdapprox*
\begin{proof}
    Consider $\Q$, the set that \GDgreedypolygonalgo returns. 
    From Lemma~\ref{lma:is_GD_partition}, $\Q$ is a \wordGDPartition of $P$ and from Lemmas~\ref{lma:GD_boundary_cover} and \ref{lma:interior_approx_part} and Corollary~\ref{cor:boundary_pieces_part}, we have $|\Q| \leq 15 \cdot \opt$.
    Since $|\Q'| \in \order{\opt}$, $|P'| \in \order{\opt \cdot n}$. 
    From Lemmas~\ref{lma:GD_greedy_runtime} and \ref{lma:constructing_interior_pieces} and Observation \ref{obs:silly_boundary} (and the discussion prior to Observation~\ref{obs:silly_boundary}) it is clear that its running time of \GDgreedypolygonalgo is dominated by the subroutine \greedyinterioralgo. The runtime guarantee follows since $|C_i| + |C_b| = |\Q| \leq 15 \cdot \opt$.
\end{proof}

\section{Open Questions}\label{sec:conclusion}
We see two major directions for further research: improving the algorithms' running times and their approximation factors. 
The running times of both \GRgreedypolygonalgo and \GDgreedypolygonalgo are dominated by the subroutine \greedyinterioralgo. 
It would be interesting to design a data structure that dynamically maintains the uncovered region (Line~\ref{line:maintain_overlay} in \greedyinterioralgo), perhaps drawing on recent work on dynamic orthogonal range searching data structures~\cite{chan_dynamic_2017,lau_algorithms_2021}. 
This could potentially reduce the overhead incurred by repeatedly updating and querying the current uncovered region and get a runtime similar to the algorithm for \GDPart in \cite{abrahamsen_partitioning_2022}. 

%In \cite[Theorem~10]{de_berg_clique-based_2023}, de Berg, Kisfaludi-Bak, Monemizadeh, and Theocharous show that a set of geodesic balls inside a polygon $P$ can be perturbed so that the resulting set forms a collection of \textit{pseudo-disks}. If \greedyinterioralgo can be modified so that the geodesic balls satisfy this pseudo-disk property, then the complexity of the uncovered region drops to $\order{|C_i|}$~\cite{kedem_union_1986}, thereby significantly improving the running time. \reilly{This is false. The union complexity just means the boundary doesn't use many sections of the same ball. This gives no guarantee on the bit complexity of these sections. This should depend on n regardless, think of the caterpillar example, this can be made to work for the minkowski sum or the boundary centers.}

Turning to approximation guarantees, we note work by Biniaz, Liu, Maheshwari, and Smid~\cite[Section~2]{biniaz_approximation_2017}, where a $4$-approximation is given for covering a \emph{discrete} set of points in $\reals^2$ using unit Euclidean disks. 
Their argument relies on the fact that a Euclidean half-ball of radius $2$ can be covered by four unit balls. 
In our setting, however, an analogous statement does not hold for geodesic unit balls (see Figure~\ref{fig:half_disks}), which presents a difficulty to improving our bounds for \GRCover and \GRPart via similar techniques.
\begin{figure}
    \centering
    \includegraphics[width=0.35\textwidth,page=7]{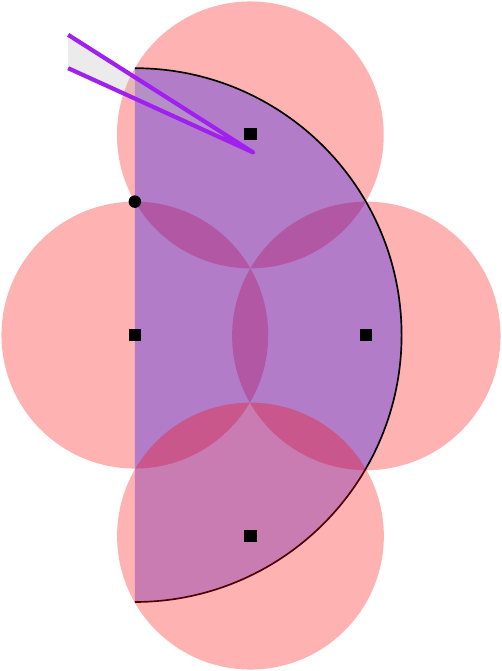}
    \caption{Four Euclidean unit balls (in red; centers marked with squares) cover a Euclidean half-ball of radius $2$ (in blue). However, the half-ball cannot be covered by four geodesic unit balls; the point, marked with a disk, is not covered by the geodesic unit balls centered at these points, even though it is distance $>\frac{1}{2}$ from the boundary. The boundary of the polygon is given in purple and the gray region is outside the polygon.}
    \label{fig:half_disks}
\end{figure}

To improve the approximation guarantee for \GDPart, we attempted to extend the approach used in Lemma~\ref{lma:interior_approx}, but this did not yield a substantially better bound without introducing significant additional complexity. By Jung's theorem~\cite{jung_ueber_1901,rosenzweig_hellys_2013}, any region of diameter $1$ is contained in a ball of radius $\frac{1}{\sqrt{3}}$. A natural idea is to partition such a ball into sectors of diameter at most $\frac{1}{2}$. 
However, this requires at least eight sectors: the regular $7$-gon inscribed in a circle of radius $\frac{1}{\sqrt{3}}$ has side length $\frac{2}{\sqrt{3}}\sin(\frac{\pi}{7}) \approx 0.501 > \frac{1}{2}$. 
This already yields a lower bound of $7 \cdot \opt$ pieces (see the proof of Theorem~\ref{thm:GD_approx} in Section~\ref{sec:diameter}).
Moreover, this only accounts for the angular width of the sectors; their radial extent still exceeds $\frac{1}{2}$, necessitating further subdivision and incurring at least one additional piece. 
Altogether, these arguments suggest that achieving a bound significantly below $9 \cdot \opt$ via such arguments is unlikely. 
In recent work we (along with several other colleagues) design a PTAS for covering and partitioning polygonal domains with pieces that fit inside a unit square \cite[Corollary 9]{aamand_covering_2026}.
It seems highly non-trivial to adapt this PTAS for the problems discussed here.
Finally, we note that the existence of an approximation algorithm for \GDCover\ remains open. Whether our methods (specifically, a suitable analogue of Theorem~\ref{thm:coverispartition} or Lemma~\ref{lma:GD_boundary_cover}) can be extended to this problem warrants further investigation.

\bibliography{bib2doi}

\end{document}